\documentclass[conference,10pt]{IEEEtran}
\IEEEoverridecommandlockouts

\usepackage{cite}
\usepackage{amsmath,amssymb,amsfonts}
\usepackage{amsthm}
\usepackage{subcaption}
\usepackage{algorithmic}
\usepackage{graphicx}
\usepackage{textcomp}
\usepackage{xcolor}
\usepackage{mathtools}
\usepackage{cleveref}

\newtheorem{definition}{Definition}

\newtheorem{remark}{Remark}
\newtheorem{theorem}{Theorem}
\newtheorem{proposition}{Proposition}
\newtheorem{lemma}{Lemma} 
\newtheorem{corollary}{Corollary}

\newcommand{\supp}{\mathrm{supp}}

\newcommand{\cA}{\mathcal{A}}
\newcommand{\cB}{\mathcal{B}}

\newcommand{\cF}{\mathcal{F}}

\newcommand{\cP}{\mathcal{P}}

\newcommand{\cX}{\mathcal{X}}
\newcommand{\cY}{\mathcal{Y}}

\newcommand{\EE}{\mathbb{E}}

\newcommand{\PP}{\mathbb{P}}

\newcommand{\RR}{\mathbb{R}}

\numberwithin{equation}{section}

\def\BibTeX{{\rm B\kern-.05em{\sc i\kern-.025em b}\kern-.08em
    T\kern-.1667em\lower.7ex\hbox{E}\kern-.125emX}}

\begin{document}

\title{ Information Contraction under 
\\ $(\varepsilon,\delta)$-Differentially Private Mechanisms \\ 
\thanks{TN acknowledges support from the IQUIST Postdoctoral Fellowship from
the Illinois Quantum Information Science and Technology Center at
the University of Illinois Urbana-Champaign.
IG was supported by the Ministry of Education, Singapore, through grant T2EP20124-0005. This project is supported by the National Research Foundation, Singapore under the NRF Postdoctoral award. CH received funding by the Deutsche Forschungsgemeinschaft (DFG, German
Research Foundation) – 550206990. This work was supported, in part, by the Federal Ministry of Research, Technology and Space (BMFTR), Germany, under the QC service center QUICS (grant no. 13N17418).}
}

\author{\IEEEauthorblockN{Theshani Nuradha}
\IEEEauthorblockA{\textit{Dept. of Mathematics} \\
\textit{University of Illinois Urbana-Champaign}\\
Urbana, IL 61801, USA \\
nuradha@illinois.edu}
\and
\IEEEauthorblockN{Ian George}
\IEEEauthorblockA{\textit{Centre for Quantum Technologies}\\
\textit{National University of Singapore}\\
Singapore 117543, Singapore \\
qit.george@gmail.com}
\and
\IEEEauthorblockN{Christoph Hirche}
\IEEEauthorblockA{\textit{Institute for Information Processing } \\
\textit{Leibniz Universit\"at Hannover}\\
Germany\\
christoph.hirche@gmail.com}
}

\maketitle

\begin{abstract}
The distinguishability of output probability distributions after being processed by a private mechanism has been a useful tool in studying various statistical and operational tasks while ensuring privacy.  To this end, standard data-processing inequalities and strong data-processing inequalities (SDPI) are employed. Most of the previously known and even tight characterizations of contraction of information measures, including total variation distance, hockey-stick divergences, and $f$-divergences, are applicable for $(\varepsilon,0)$-local differentially private (LDP) mechanisms. In this work, we derive both linear and non-linear strong data-processing inequalities for hockey-stick divergence and $f$-divergences that are valid for all $(\varepsilon,\delta)$-LDP mechanisms even when $\delta \neq 0$. Our results either generalize or improve the previously known bounds on the contraction of these distinguishability measures. 

\end{abstract}


\section{Introduction}




Ensuring the privacy of data collected for various inference tasks is critical due to 
the rapid development of information processing techniques and technologies. To this end, statistical privacy frameworks provide provable privacy guarantees~\cite{DMNS06,DR14,KM14, CY16, mironov2017renyi, nuradha2022pufferfishJ}. Local differential privacy (LDP) is one such model (\Cref{def: LDP}), where individual data records are kept private while answering aggregate queries~\cite{erlingsson2014rappor}.  

Studying statistical inference under LDP constraints is vital for understanding the price that we have to pay to ensure privacy, especially in distributed settings. 
To this end, the contraction of divergences under LDP constraints is an important technical tool. For example, it was shown  in~\cite{kairouz2014extremal} that, for a mechanism $A$ satisfying  $(\varepsilon,0)$-LDP,
the total variation distance $ \operatorname{TV}(P,Q) \coloneqq \frac{1}{2}\left\| P-Q\right\|_1 $ between two probability distributions $P$ and $Q$ satisfies 
  $ \operatorname{TV}(A(P),A(Q) ) \leq  \frac{e^\varepsilon -1}{e^\varepsilon +1} \operatorname{TV}(P,Q)$.
Bounds for the contraction of other divergences, including hockey-stick, chi-square, Hellinger, Kullback–Leibler divergence, and more generally for $f$-divergences have been studied in~\cite{DJW13,duchi2018minimax,Contraction_local_new24,zamanlooy2023strong, zamanlooy2024mathrm}, mainly for the setting $\delta =0$. 

Moreover, the hockey-stick divergence has an operational equivalence to the differential privacy framework. As such, studying this divergence provides new operational insights on statistical tasks under privacy~\cite{asoodeh2021local}. Furthermore, the hockey-stick divergence finds relevance in binary hypothesis testing as it quantifies the error probability 
when we have unequal priors between the probability distributions~\cite{Polyanskiy_Wu_2025}. 

In this work, we study the information contraction of hockey-stick and $f$-divergences under $(\varepsilon,\delta)$-LDP without assuming $\delta = 0$. In Section \ref{sec:contraction-of-HS-div}, we establish linear (Proposition \ref{prop:contraction_coeff_upper_bound}) and non-linear SDPI (Theorem \ref{thm:non_linear_HS_div}) for hockey-stick divergences. We show our results generalize those of \cite{zamanlooy2023strong,zamanlooy2024mathrm} as we recover those works' results by setting $\delta$ equal to zero. Also, we show that in the regime $1<\gamma<e^\varepsilon$, the non-linear SDPI is strictly tighter than the linear SDPI whenever the input distinguishability lies in $(0,1)$ (\Cref{cor:nonlinear-better}). Moreover, we show that the non-linear SDPI in Theorem \ref{thm:non_linear_HS_div} can be achieved by a private binary mechanism for certain parameter regimes. 
In Section \ref{sec:contraction-of-f-div}, we use our linear SDPI bounds to obtain upper bounds on $f$-divergences relative to total variation distance (Proposition \ref{prp:f_div_Contraction}) for the setting $\delta \in[0,1]$. We then highlight the generality and tightness of Proposition \ref{prp:f_div_Contraction} by showing it recovers a result of \cite{zamanlooy2023strong} in the case $\delta = 0$ and that for $\delta \neq 0$ it empirically is a better bound than a similar bound derived in concurrent work \cite{dasgupta2025quantum}. In summary, building on the known properties of hockey-stick divergence together with interpolations and truncations to incorporate $\delta >0$ case, our work provides upper bounds on the contraction of hockey-stick divergence and $f$-divergences that, to the best of our knowledge, were not previously known for $\delta \in (0,1]$. 







\section{Preliminaries and Notation}
We denote by $(\Omega,\cF,\PP)$ the underlying probability space on which all random variables (RVs) are defined. 
RVs are denoted by uppercase letters, e.g. $X$, with $P_X$ representing the corresponding probability law (we sometimes omit the subscript $X$ for clarity). For $X\sim P_X$, we use $\supp(X)$ for the support. The joint law of $(X,Y)$ is denoted by $P_{XY}$, while $P_{Y|X}$ represents the conditional probability of $Y$ given $X$. We denote the set of probability distributions supported on the space $\cX$ by $\cP(\cX)$.

The hockey-stick divergence is defined for $\gamma \geq 1$ and $P,Q \in \cP(\cX)$ as $ E_\gamma (P\| Q):=\sup_{\cB} P(\cB) - \gamma Q(\cB), $
where the supremum is over all measurable sets $\cB \subseteq \cX$. The total variation (TV) distance is the special case when $\gamma =1$. For finite $\cX$, we have that $E_\gamma(P \Vert Q)= \sum_{x \in \cX} \max\{0, P(x)-\gamma Q(x) \}.$

The contraction coefficient of the hockey-stick divergence for $\gamma \geq 1$ and channel $P_{Y|X}$ is defined as
\begin{equation}
    \eta_{\gamma
    }(P_{Y|X}) \coloneqq \sup_{\substack{P, Q \textnormal{ s.t.} \\ E_\gamma(P \Vert Q) \neq 0}} \frac{  E_{\gamma}\!\left( P_{Y|X} \circ P \Vert P_{Y|X} \circ Q \right) }{E_\gamma(P \Vert Q)}.
\end{equation}
It was shown in~\cite[Proposition~5.9]{cohen1998comparisons} (for finite alphabets) and \cite[Theorem~2]{privacyAmplificationHS20} that 
\begin{equation} \label{eq:contraction_coeff_simpl_HS}
     \eta_{\gamma
    }(P_{Y|X}) = \sup_{x,x' \in \cX} E_{\gamma}\!\left( P_{Y|X} (\cdot| x) \Vert P_{Y|X} (\cdot |x' ) \right). 
\end{equation}
We also utilize the following properties of the hockey-stick divergence in our work: For $\gamma \geq 1$ and $1 \leq \beta= 1- (1-\gamma) /E_{\gamma}(P \Vert Q) $ ~\cite{balle2019privacy} (see also~\cite[Eq.(8)]{zamanlooy2024mathrm})
\begin{align} \label{eq:HS_different_gamma_contraction}
    & E_{\gamma}\!\left( P_{Y|X} \circ P \Vert P_{Y|X} \circ Q  \right) \leq \eta_{\beta}(P_{Y|X})\, E_{\gamma}(P \Vert Q)
\end{align}
and for $1 \leq \gamma \leq \gamma'$~\cite[Proposition~1]{zamanlooy2024mathrm}, 
    \begin{equation} \label{eq:HS_different_gamma}
        E_{\gamma}(P\Vert Q) \leq \frac{\gamma'- \gamma}{\gamma'+1} + \frac{\gamma+1}{\gamma'+1} \max\{ E_{\gamma'}(P \Vert Q), E_{\gamma'}(Q \Vert P)\}.
    \end{equation}
    
 The $\max$-divergence is defined as 
\begin{equation}
    D_{\max}( P \Vert Q) \coloneqq \log \sup_{ \substack{S \in \mathrm{Supp}(Y) \\  Y \sim P, Z \sim Q}} \frac{\PP[Y \in S]}{\PP[Z \in S]}.
    \end{equation}
  Smooth $\max$-divergence (aka approximate-max divergence~\cite{DR14}) is defined as 
\begin{equation}
    D_{\max}^\delta( P \Vert Q) \coloneqq \log \sup_{S \in \mathrm{Supp}(Y), \PP[Y \in S] \geq \delta} \frac{\PP[Y \in S] -\delta}{\PP[Z \in S]},
    \end{equation}
    where $Y$ and $Z$ are RVs distributed according to $Y \sim P$ and $Z \sim Q$.
    Also, note that $D_{\max}^\delta$ has the following duality with the hockey-stick divergence~\cite{liu2016e} (see also~\cite[Remark~4]{nuradha2023quantum}):
    \begin{equation} \label{eq:duality}
        D_{\max}^\delta( P \Vert Q)= \log\inf\left\{ \lambda \geq 0: E_\lambda(P\Vert Q) \leq \delta \right\} .
    \end{equation}

The Kullback-Leibler (KL) divergence between $P, Q\in\cP(\cX)$ with $P\ll Q$ (absolutely continuous) is $ D_{\operatorname{KL}}(P \Vert Q) \coloneqq \EE_P\left[\log \left(\frac{d P}{d Q}\right)\right],$
where $\frac{d P}{d Q}$ is the Radon-Nikodym derivative of $P$ with respect to (w.r.t.) $Q$ and $\EE_P$ represents the expectation w.r.t. $P$. For finite $\cX$, we have that $D_{\operatorname{KL}}(P \Vert Q) = \sum_{x \in \cX} P(x) \log\!\left( \frac{P(x)}{Q(x)}\right).$

Let $f: (0,\infty) \to \RR$ be a convex and twice differentiable function 
satisfying $f(1)=0$. Then the $f$-divergence between $P, Q\in\cP(\cX)$ with $P\ll Q$ is defined  $D_f(P\|Q) \coloneqq \EE_Q\left[ f\left( \frac{dP}{dQ}\right) \right]$. 
For finite $\cX$, we have $ D_f(P\|Q) = \sum_{x \in \cX} Q(x)\, f\left( \frac{P(x)}{Q(x)} \right)$. Also note that $D_{\operatorname{KL}}$ is equivalent to $D_{f}$ with $f=x \log( x)$.
Sason and Verd\'u gave an alternative representation \\ for $f$-divergences~\cite[Proposition 3]{sason2016f}. For two discrete probability mass functions $P$ and $Q$ with equal support on a set $\mathcal{X}$, we have,
\begin{equation}\label{eq:f_divergence}
    D_f(P \Vert Q) \coloneqq \hspace{-1mm} \int_{1}^{\infty} \hspace{-2mm} f''(\gamma) E_\gamma(P \Vert Q) + \gamma^{-3} f''(\gamma^{-1}) E_\gamma(Q \Vert P) \ \mathrm{d} \gamma.
\end{equation}

Next, we define local differential privacy (LDP) and state equivalent formulations for the LDP framework.

\begin{definition}[Local Differential Privacy] \label{def: LDP}
Fix $\varepsilon \geq 0$ and $\delta \in [0,1]$. A~randomized mechanism ${A}: \cX \to \cY$  (a conditional probability distribution $P_{Y|X}$) is $(\varepsilon,\delta)$-local differentially private if 
\begin{equation}
\PP\big(A(x) \in \cB \big) \leq e^{\varepsilon} \hspace{1mm} \PP\big(A(x') \in \cB\big) + \delta,\label{eq:dp_def}
\end{equation}\\[-3.5mm]
for all $x,x' \in \cX$ and $\cB \subseteq \cY $ measurable. 
\end{definition}
From~\cite[Theorem~1]{asoodeh2021local} and~\eqref{eq:contraction_coeff_simpl_HS}, we have that a mechanism $A$ satisfying $(\varepsilon, \delta)$-LDP is equivalent to the following statements:
\begin{align}
& \sup_{P, Q \in \cP(X)} E_{e^\varepsilon}\!\left( A(P) \Vert A(Q) \right) \leq \delta \\ 
  &\iff   \eta_{e^\varepsilon}(A) \leq \delta \\
  &\iff 
\sup_{P, Q \in \cP(X)}   D_{\max}^\delta \!\left( A(P) \Vert A(Q) \right) \leq \varepsilon,
\end{align}
where the last implication follows from~\cite{DR14} along with the duality between smooth max-divergence and hockey-stick divergence.

\subsection{Previous results}
The following result is also relevant to our discussion.
 \begin{proposition}[Proposition 2,~\cite{zamanlooy2024mathrm}]
 \label{Prop:gamma-gamma-ev} Let $\cX$ be a finite set and $P, Q \in \cP(\cX)$.
     Set $\gamma\geq\gamma'\geq1$. If we have  \\ $  E_{\gamma'}(P \Vert Q) \leq (\gamma-\gamma')\, \min_{x \in \cX} Q(x), $
     then $E_\gamma(P\|Q) = 0.$
 \end{proposition}
 At this point, we state two quick corollaries of this result.
 \begin{corollary} \label{Cor:FiniteX_D_max_HS}
 Let $\cX$ be a finite set and $P, Q \in \cP(\cX)$.
    For $\gamma\geq1$, we have $ D_{\max}(P \Vert Q) \leq \log\left( \gamma + \frac{E_{\gamma}(P \Vert Q)}{\min_{x \in \cX} Q(x)} \right).$
\end{corollary}
\begin{proof}
    Fix some $\gamma'\geq1$. Say $E_{\gamma'}(P \Vert Q) = x$. Then there exists a $\gamma\geq\gamma'$ such that $x=(\gamma-\gamma')\, \min_{x} Q(x)$. It follows that $E_\gamma(P \Vert Q)=0$ by~\Cref{Prop:gamma-gamma-ev}. We also have $\gamma = \frac{x}{\min_{x} Q(x)}+\gamma'$. However, we also know that the first point at which the hockey-stick divergence becomes zero is given by the 
    max-divergence. Hence, the above gives an upper bound as stated in the claim. 
\end{proof}
The special case of~\Cref{Cor:FiniteX_D_max_HS} for $\gamma=1$ can be found in~\cite{sason2015reverse_arXiv}. 
As a direct consequence, one can also prove the following statement, which will be useful later.
\begin{corollary}\label{Cor:Dmax-by-smooth} Let $\cX$ be a finite set and $P, Q \in \cP(\cX)$.
    For $0\leq\delta\leq1$, we have
    \\
    $D_{\max}(P \Vert Q) \leq \log\left( e^{D^\delta_{\max}(P \Vert Q)} + \frac{\delta}{\min_{x \in \cX} Q(x)} \right).$
\end{corollary}
\begin{proof}
    This follows from the duality between the smooth max-divergence and the hockey-stick divergence given in~\eqref{eq:duality}. 
\end{proof}
 
\section{Contraction of Hockey-Stick Divergences}\label{sec:contraction-of-HS-div}

In this section, we study the contraction of the hockey-stick divergence by deriving linear bounds on the contraction coefficients
and non-linear strong data processing inequalities. 
Note that all the contraction bounds mentioned in this section are valid for arbitrary $\cX$ and $\cY$ with $A:\cX \to \cY$ satisfying $(\varepsilon,\delta)$-LDP.

First, by the data-processing inequality (DPI) of the hockey-stick divergence, we have that for $A: \cX \to \cY$ that satisfies $(\varepsilon,\delta)$-LDP, $\gamma \geq 1$, and $P, Q \in \cP(\cX)$
\begin{equation} \label{eq:DPI_E_gamma}
     E_{\gamma}\!\left( A(P) \Vert A(Q) \right) \leq E_{\gamma}(P \Vert Q).
\end{equation}
Next, we show that we can strengthen this data-processing inequality via strong data-processing inequalities (SDPI) when the mechanism satisfies $(\varepsilon,\delta)$-LDP.
\begin{proposition} [Linear SDPI] \label{prop:contraction_coeff_upper_bound}
    Let $A: \cX \to \cY$ satisfy $(\varepsilon,\delta)$-LDP and $1 \leq \gamma \leq e^\varepsilon$. Then, we have that for $P,Q \in \cP(\cX)$
    \begin{equation} \label{eq:contraction_HS_numerator}
        E_{\gamma}\!\left( A(P) \Vert A(Q) \right) \leq \frac{(e^\varepsilon- \gamma) +\delta (\gamma+1) }{e^\varepsilon+1} .
    \end{equation}
    Furthermore, we have that for $e^\varepsilon \geq \gamma \geq 1$
    \begin{equation}\label{eq:bound-on-HS-contraction-coeff}
        \eta_{\gamma}(A) \leq \frac{(e^\varepsilon- \gamma) + \delta (\gamma+1) }{e^\varepsilon+1} .
    \end{equation}
For $\gamma \geq e^\varepsilon$, we have that $\eta_{\gamma}(A) \leq \delta.$
This results in, for $\gamma \geq 1$, 
\begin{equation}
        \eta_{\gamma}(A) \leq \max \left\{  \frac{(e^\varepsilon- \gamma) + \delta (\gamma+1) }{e^\varepsilon+1} , \delta \right\}.
    \end{equation}
\end{proposition}
\begin{proof}
 Since $A$ satisfies $(\varepsilon,\delta)$-LDP, we have that $ \max\left \{ E_{e^\varepsilon}(A(P) \Vert A(Q)),  E_{e^\varepsilon}(A(Q) \Vert A(P))\right\} \leq \delta.$
    Then, by applying~\eqref{eq:HS_different_gamma} with $\gamma'= e^\varepsilon$ we get 
    \begin{align*}
       &  E_{\gamma}(A(P)\Vert A(Q))  \leq \frac{e^\varepsilon- \gamma}{e^\varepsilon+1}  \nonumber \\ 
        & \quad + \frac{\gamma+1}{e^\varepsilon+1} \max\{ E_{e^\varepsilon}(A(P) \Vert A(Q)), E_{e^\varepsilon}(A(Q) \Vert A(P))\},
    \end{align*}
     and together with that we have that 
         $E_{\gamma}\!\left( A(P) \Vert A(Q) \right) \leq \frac{(e^\varepsilon- \gamma) +\delta (\gamma+1) }{e^\varepsilon+1}$.
    For the contraction coefficient, recall that
       $\eta_{\gamma}(A) = \sup_{x, x' \in \cX} E_{\gamma}( A(\cdot| x) \Vert A(\cdot| x')).$ 
    This immediately concludes the proof by applying the first proposition statement for $e^\varepsilon \geq \gamma \geq 1$. For $\gamma \geq e^\varepsilon $, it follows as above together with the monotonicity of hockey-stick divergence.

    For the last inequality, consider the two inequalities that we proved for those two separate regimes. See that 
         $\frac{(e^\varepsilon - \gamma) + \delta (\gamma+1) }{e^\varepsilon +1}  - \delta = \frac{(1-\delta) (e^\varepsilon  -\gamma)}{(e^\varepsilon  +1)}$. 
    From this, it is evident that the first bound is strictly larger than $\delta >0$ if and only if $e^\varepsilon  > \gamma$. This leads to the combination of both the regimes, concluding the proof.    
\end{proof}
\begin{remark}
In \Cref{prop:contraction_coeff_upper_bound}, we study linear SDPI for the hockey-stick divergence with $\gamma \geq 1$, $\varepsilon \geq 0$ and $\delta \in [0,1]$. 
For $\gamma \geq 1$ and $\delta =0$, it recovers~\cite[Theorem 1]{zamanlooy2023strong}. 
In the case $\delta \in [0,1]$ and $\gamma=1$, \Cref{prop:contraction_coeff_upper_bound} and the quantum result \cite[Theorem~5]{nuradha2024contraction} evaluated on classical distributions give the same bound (See also~\cite[Proposition~8]{nuradha2025MeasuredHS} and~\cite[Proposition~10]{nuradha2023quantum} for other quantum proofs that specialize to the classical setting).
For $\delta \in [0,1]$ and $\gamma \geq 1$, an independent and concurrent work \cite{dasgupta2025quantum} establishes the same result with a different proof technique. 
\end{remark}

Note that~\Cref{prop:contraction_coeff_upper_bound} implies that for $\gamma \geq 1$
\begin{multline}
      E_{\gamma}\!\left( A(P) \Vert A(Q) \right)  \leq \\ \max \left\{  \frac{(e^\varepsilon- \gamma) + \delta (\gamma+1) }{e^\varepsilon+1} , \delta \right\} E_{\gamma}(P\Vert Q),
\end{multline}
which is the reason that we call it a linear SDPI whenever $\max \left\{  \frac{(e^\varepsilon- \gamma) + \delta (\gamma+1) }{e^\varepsilon+1} , \delta \right\} <1$. Next, we obtain a non-linear SDPI inequality, where we cannot completely isolate input and output distinguishability with a constant factor, and later show that it is tighter than the linear one for $\gamma \in (1, e^\varepsilon)$.
\begin{theorem} [Non-Linear SDPI]\label{thm:non_linear_HS_div}
    Let $A: \cX \to \cY$ satisfy $(\varepsilon,\delta)$-LDP. Then, we have that for $\gamma \geq 1$ and $P, Q \in \cP(\cX)$
    \begin{multline}
         E_{\gamma}\!\left( A(P) \Vert A(Q) \right) \leq \\  \max \Big \{\frac{(e^\varepsilon +2 \delta -1) E_{\gamma}(P \Vert Q) -(\gamma-1)(1-\delta)}{e^\varepsilon +1},  \\  \delta E_{\gamma}(P \Vert Q) \Big\}.
    \end{multline}   
\end{theorem}
\begin{proof}
First, by applying~\eqref{eq:HS_different_gamma_contraction} with $P_{Y|X}=A$ we have that 
\begin{align}
    & E_{\gamma}\!\left( A(P) \Vert A(Q) \right) \leq \eta_{\beta}(A)\, E_{\gamma}(P \Vert Q) \\
      & \leq  \max \! \left\{\frac{(\gamma- \beta) + \delta (\beta+1) }{\gamma+1}, \delta \right\} E_{\gamma}(P \Vert Q) ,
\end{align}\\[-3.5mm]
where the last inequality follows from~\Cref{prop:contraction_coeff_upper_bound}. We arrive at the desired inequality by algebraic simplifications together with the substitution of $\beta= 1-  (1-\gamma)/E_{\gamma}(P \Vert Q) $.    
\end{proof}


\begin{remark}[Achievability of SDPI]
The upper bound in~\Cref{thm:non_linear_HS_div} is achieved by two Bernoulli distributions (i.e.; $P, Q \in \cP(\cX)$ with $\cX=\{0,1\}$) and mechanism $A$ being the binary symmetric channel denoted as $A_p$ with the flipping parameter $p=(1-\delta)/(e^\varepsilon +1)$ whenever $\delta=0$ or $E_{\gamma}(P\Vert Q) \geq (\gamma-1)/(e^\varepsilon +1)$ for $\delta >0$. In particular for the above $P, Q$ and $A_p$, which satisfies $(\varepsilon,\delta)$-LDP, we have that 
$E_\gamma\!\left( A(P) \Vert A(Q) \right)= \max\{ g(P,Q, \gamma,\varepsilon, \delta), 0 \}$ where 
$g(P,Q, \gamma,\varepsilon, \delta) \coloneqq E_{\gamma}(P \Vert Q) \frac{(e^\varepsilon -1 +2 \delta)}{(e^\varepsilon +1)} +\frac{(1-\delta)}{(e^\varepsilon +1)} (1-\gamma)$.
In the special case $\delta=0$, the above was established in~\cite[Theorem~2]{zamanlooy2023strong}.
\end{remark}

\begin{corollary}\label{cor:nonlinear-better}
Assume $1<\gamma<e^\varepsilon$, $\delta\in[0,1)$, and let $E_\gamma(P\Vert Q)\in(0,1)$. Then the non-linear SDPI bound in~\Cref{thm:non_linear_HS_div} is strictly smaller than the linear SDPI bound obtained from~\Cref{prop:contraction_coeff_upper_bound}.
\end{corollary}
\begin{proof} For $1<\gamma<e^\varepsilon$, \Cref{prop:contraction_coeff_upper_bound} gives: for $t\coloneqq E_\gamma(P\Vert Q)$ linear SDPI bound $ L(t) \coloneqq \frac{(e^\varepsilon-\gamma)+\delta(\gamma+1)}{e^\varepsilon+1} t$,
while the non-linear SDPI from~\Cref{thm:non_linear_HS_div} gives \\ $N(t)\coloneqq \max \Big \{\frac{(e^\varepsilon +2 \delta -1) t -(\gamma-1)(1-\delta)}{e^\varepsilon +1},   \delta t \Big\}.$
By algebraic simplifications, we arrive at  $L(t)-N(t)= \\ \frac{1-\delta}{e^\varepsilon +1}   \min \left\{ (\gamma-1) (1-E_{\gamma} (P \Vert Q)), (e^\varepsilon- \gamma) E_\gamma(P \Vert Q) \right \}$,
concluding the proof.
\end{proof}

\begin{remark}
    In Fig.~\ref{fig:compare}, we compare the upper bounds provided by standard DPI, linear SDPI in~\Cref{prop:contraction_coeff_upper_bound}, and non-linear SDPI in~\Cref{thm:non_linear_HS_div} for a setting where $\delta >0$.
    Furthermore, \Cref{thm:non_linear_HS_div} recovers~\cite[Theorem~3]{zamanlooy2024mathrm} for $\delta =0$. 
   \end{remark}




\begin{remark}[$F_\gamma$ curves]
  Recall the definition of $F_\gamma$ curves from \cite{zamanlooy2024mathrm}: For $t \in [0,1]$ and channel $P_{Y|X}$
\begin{multline*} \label{eq:f_gamma_def}
    F_\gamma^{P_{Y|X}}(t) \coloneqq \\ 
    \sup_{P,Q} \left\{ E_\gamma\!\left( P_{Y|X} \circ P \Vert P_{Y|X} \circ Q  \right): E_\gamma (P \Vert Q) \leq t \right\}.
\end{multline*}
Let $A: \cX \to \cY$ satisfy $(\varepsilon,\delta)$-LDP. 
By applying~\Cref{thm:non_linear_HS_div} to the definition of $F_\gamma$, we have that for $\gamma \geq 1$:
\begin{equation*}
    F_{\gamma}^A(t) \leq  \max \left\{\frac{(e^\varepsilon+2 \delta -1) t -(\gamma-1)(1-\delta)}{e^\varepsilon +1}, \delta t\right\}.
\end{equation*}\\[-3.5mm]
The above result recovers \cite{zamanlooy2024mathrm} when $\delta=0$, 
and can be utilized to obtain upper bounds on the contraction of the sequential composition of several private mechanisms.
\end{remark}

\begin{figure}[ht]
    \centering
\includegraphics[width=0.9\linewidth]{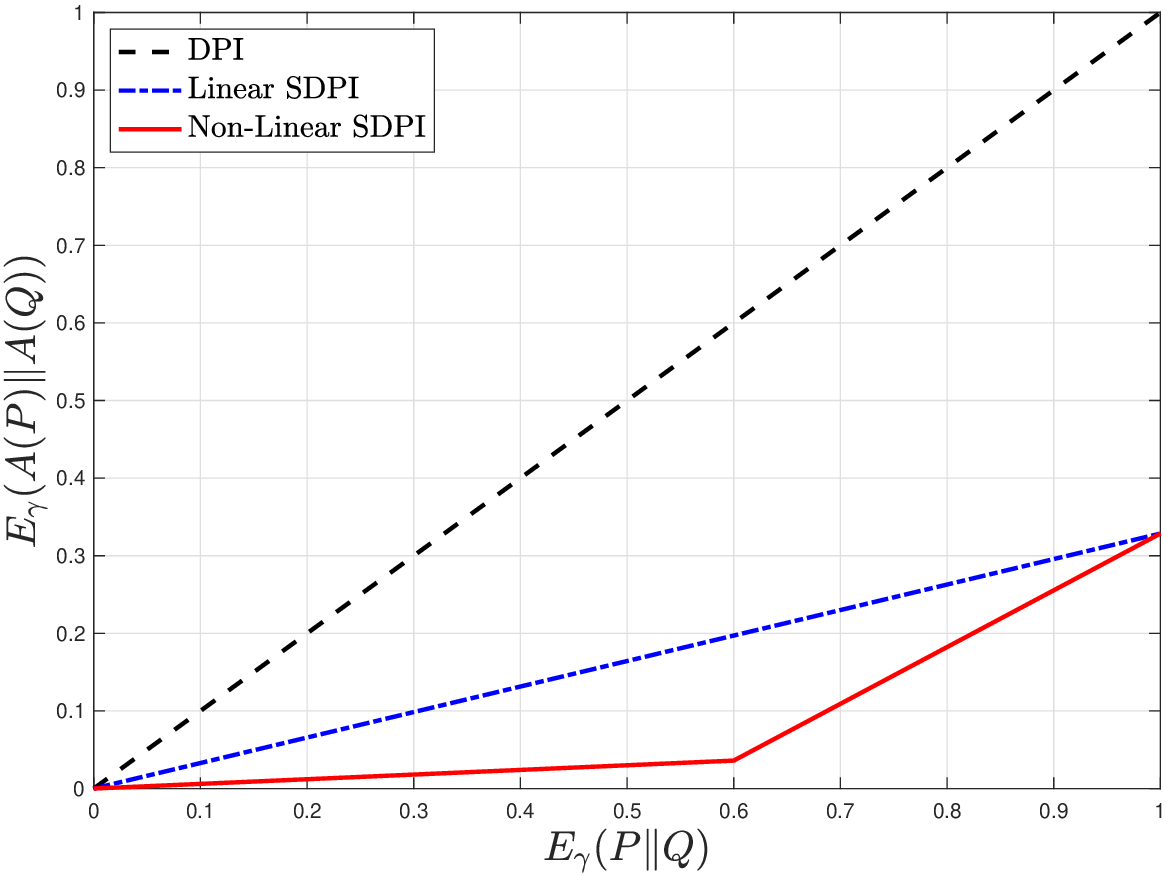} \caption{Comparison of information contraction inequalities: DPI refers to the standard data-processing inequality in~\eqref{eq:DPI_E_gamma}; Linear SDPI refers to~\Cref{prop:contraction_coeff_upper_bound}; and Non-Linear SDPI refers to~\Cref{thm:non_linear_HS_div}. In this example setting, we consider $\varepsilon= \ln(6), \gamma=4 < e^\varepsilon, \delta=0.06$ for a mechanism $A$ satisfying $(\varepsilon,\delta)$-LDP. Each of these lines/curves shows the largest $E_{\gamma}\!\left(A(P) \Vert A(Q)\right)$ value that can be reached for the input distinguishability $E_{\gamma}(P \Vert Q) \in [0,1]$. 
}
    \label{fig:compare}
\end{figure}

\vspace{-3mm}
\section{Bounds on $f$-Divergences}\label{sec:contraction-of-f-div}
In this section, we study the contraction bounds of $f$-divergences relative to TV distance when processed by an $(\varepsilon, \delta)$-LDP mechanism. These results are valid even when $\delta >0$. The contraction bounds for $f$-divergences in this section are valid for finite $\cY$ when $\cA: \cX \to \cY$. 
First, we need the following lemma, which is a slight extension of previous results in~\cite{zamanlooy2023strong,hirche2024quantumDivergences}.

\begin{lemma}\label{Lem:HS-Bound-HS}
    Let $1\leq\gamma_1\leq\gamma\leq\gamma_2$, then
    \begin{align*}
        E_\gamma(P\|Q) \leq &\frac{\gamma-\gamma_2}{\gamma_1-\gamma_2} E_{\gamma_1}(P\|Q) 
        + \frac{\gamma_1-\gamma}{\gamma_1-\gamma_2}E_{\gamma_2}(P\|Q). 
    \end{align*}
\end{lemma}
\begin{proof}
    The proof follows from the convexity of the function $\gamma\to E_\gamma$. As a result for any $1\leq\gamma_1\leq \gamma \leq \gamma_2$, $E_\gamma$ is upper bounded by the straight line connecting the two corresponding points of the function, which is given by $
\frac{\gamma_2-\gamma}{\gamma_2-\gamma_1}E_{\gamma_1}(P\Vert Q)
+
\frac{\gamma-\gamma_1}{\gamma_2-\gamma_1}E_{\gamma_2}(P\Vert Q)$. Rewriting the coefficients, this line is the claimed result. 
\end{proof}
We can now give a bound on general $f$-divergences with the assumption $\inf_{y \in \cY} A(Q)(y) >0$. 
\begin{proposition} \label{prp:f_div_Contraction}
    Let $A:\cX \to \cY$ be an $(\varepsilon,\delta)$-LDP mechanism (given as $P_{Y|X}$) such that $\lambda := \inf_{y \in \cY} A(Q)(y) >0$, then
    \begin{align}
       & D_f(A(P)\|A(Q)) \leq  \nonumber\\
    &\frac{f(e^\varepsilon)+e^\varepsilon f(e^{-\varepsilon})}{e^\varepsilon-1}\frac{e^\varepsilon-1+2\delta}{e^\varepsilon+1}\tau - \frac{f(e^\varepsilon)+f(e^{-\varepsilon})}{e^\varepsilon-1}\delta \nonumber\\
        &+ \lambda \Big( f\left(e^\varepsilon+\frac{\delta}{\lambda}\right) - f\left(e^\varepsilon\right) + \left(e^\varepsilon+\frac{\delta}{\lambda}\right)f\left(\left(e^\varepsilon+\frac{\delta}{\lambda}\right)^{-1}\right) \nonumber\\
        & \quad \quad \quad -\left(e^\varepsilon+\frac{\delta}{\lambda}\right)f\left(e^{-\varepsilon}\right) \Big),
    \end{align}
    where $\tau=E_1(P\|Q)$ is the TV distance.  
\end{proposition}
\begin{proof}
    We start with the integral representation in Eq.~\eqref{eq:f_divergence}. The start of the proof is similar to that of~\cite[Proposition 5.2]{hirche2024quantumDivergences}. Crucial is that the hockey-stick divergence $E_\gamma(P\|Q)$ is zero for all $\gamma\geq\exp(D_{\max}(P\|Q))$. This allows us to limit the range of the integral. Then, we split the integral into two parts, $1\leq\gamma\leq e^{\varepsilon}$ and $e^{\varepsilon}\leq\gamma\leq e^{D_{\max}(A(P)\|A(Q))}$. For $1\leq\gamma\leq e^{\varepsilon} $, using~\Cref{Lem:HS-Bound-HS}, we have 
    $E_\gamma(A(P) \Vert A(Q)) \leq \frac{e^\varepsilon-\gamma}{e^\varepsilon-1} E_1(A(P) \Vert A(Q)) +\frac{e^\varepsilon-\gamma}{e^\varepsilon-1}E_{e^\varepsilon} (A(P) \Vert A(Q))$. Since $A$ is $(\varepsilon,\delta)$-LDP, we have $E_{e^\varepsilon} (A(P) \Vert A(Q)) \leq \delta$ and by~\Cref{prop:contraction_coeff_upper_bound}, $E_1(A(P) \Vert A(Q)) \leq \frac{e^\varepsilon-1+2\delta}{e^\varepsilon+1} E_1(P \Vert Q) $. We get the same upper bound for $E_\gamma(A(Q) \Vert A(P))$ by switching the roles of $P$ and $Q$. 
    By computing the integral in the range $\gamma \in [1,e^\varepsilon]$ with the remaining terms, we obtain the first and the second term of the upper bound. 

    Now, considering the second part of the integral, with $\gamma \in [e^\varepsilon,e^{D_{\max}(A(P)\|A(Q))}] $, we use \Cref{Lem:HS-Bound-HS} again, now with $\gamma_1=e^\varepsilon$ and $\gamma_2=e^{D_{\max}(A(P)\|A(Q))}$. Since $E_{\gamma_2}(A(P)\|A(Q)) =0$ and $E_{\gamma_1}(A(P)\|A(Q)) \leq \delta $, we get  $E_\gamma(A(P) \Vert A(Q)) \leq \frac{e^{D_{\max}(A(P)\|A(Q))}-\gamma}
     {e^{D_{\max}(A(P)\|A(Q))}-e^\varepsilon}
\delta $. Similarly,  we get  $E_\gamma(A(Q) \Vert A(P)) \leq \frac{e^{D_{\max}(A(P)\|A(Q))}-\gamma}
     {e^{D_{\max}(A(P)\|A(Q))}-e^\varepsilon}
\delta $.  Using \Cref{Cor:Dmax-by-smooth} and recalling that for $A$ satisfying $(\varepsilon,\delta)$-LDP, $D_{\max}^\delta(A(P)\|A(Q)) \leq \varepsilon$ , we obtain 
$D_{\max}(A(P)\|A(Q)) \leq \log\left(e^{D_{\max}^\delta(A(P)\|A(Q))}+\frac{\delta}{\inf_{y\in\cY}A(Q)(y)}\right) \leq \log\left(e^{\varepsilon}+\frac{\delta}{\inf_{y\in\cY}A(Q)(y)}\right)$. Together with these inequalities, by computing the integral within $\gamma \in [e^\varepsilon,e^{D_{\max}(A(P)\|A(Q))}] $, we arrive at the remaining terms of the claimed upper bound, concluding the proof.
\end{proof}

Note that by choosing $\lambda=\inf_{x\in \cX, y \in \cY} P_{Y|X}(y|x) $, we obtain a similar bound to the above that is independent of the input distributions.

\Cref{prp:f_div_Contraction} provides a very general result. To get a better idea, we will compare some special cases. First note that for $\delta=0$, we recover 
    \begin{align*}
        D_f(A(P)\|A(Q)) \leq &\frac{f(e^\varepsilon)+e^\varepsilon f(e^{-\varepsilon})}{e^\varepsilon+1} \ E_1(P\|Q),
    \end{align*}
which was previously shown in~\cite[Theorem~5]{zamanlooy2023strong}.

For further comparison, it will be helpful to specialize to KL divergence.
\begin{corollary}\label{Cor:RE-LDP-bound}
       Let $A:\cX \to \cY$ be an $(\varepsilon,\delta)$-LDP mechanism (represented as $P_{Y|X}$) such that $\lambda := \inf_{y \in \cY} A(Q)(y) >0$,
       then
    \begin{align}
      &  D_{\operatorname{KL}}(A(P)\|A(Q)) \leq \varepsilon\frac{e^\varepsilon-1+2\delta}{e^\varepsilon+1}\tau - \varepsilon\frac{e^\varepsilon-e^{-\varepsilon}}{e^\varepsilon-1}\delta \nonumber\\
        &+ \lambda \left( \left(e^\varepsilon+\frac{\delta}{\lambda}-1\right)\log\left(e^\varepsilon+\frac{\delta}{\lambda}\right) + \left(1-e^\varepsilon+\frac{\delta}{\lambda}e^{-\varepsilon}\right)\varepsilon \right),\label{Eq:RE-LDP-bound}
    \end{align} \\[-4.5mm]
    where $\tau=E_1(P\|Q)$ is the TV distance. 
\end{corollary}
\begin{proof}
    This follows by choosing $f(x)=x\log(x)$ in~\Cref{prp:f_div_Contraction} and rearranging the result.
\end{proof}
There are only a few bounds in the literature that apply to $\delta>0$. A recent point of comparison can be found in~\cite[Theorem 6]{dasgupta2025quantum}, which shows for classical probability distributions $P$ and $Q$ that, 
\begin{align}
&D_{\operatorname{KL}}(A(P)\|A(Q)) \leq  \nonumber \\ &\left( \varepsilon\tanh\!\left(\frac\varepsilon2 \right) + \delta\left(\frac{2\varepsilon}{e^\varepsilon+1}+\frac{e^\varepsilon+\delta-1}{e^\varepsilon} + \log\frac1{1-\delta}\right)\right)\tau \ + \nonumber\\
    & \delta\left( \frac{e^\varepsilon}{1-\delta} + 2\log\frac{e^\varepsilon}{1-\delta} - \frac{1-\delta}{e^\varepsilon} + 2\left( \varepsilon + \log\frac1{1-\delta} +\frac2m \right) \right), \label{Eq:dasgupta}
\end{align}\\[-3.5mm]
where $m\equiv m(A,Q,P)$ involves a truncated distribution defined in~\cite{dasgupta2025quantum}. We always have $m\leq\lambda$ as a point of comparison. For better comparison, we can write out the $\tanh$ function and reformulate~\eqref{Eq:dasgupta} as
\begin{align}
   & D_{\operatorname{KL}}(A(P)\|A(Q)) \leq  \nonumber \\ 
   &\left( \varepsilon\frac{e^\varepsilon-1+2\delta}{e^\varepsilon+1} + \delta\left(\frac{e^\varepsilon+\delta-1}{e^\varepsilon} + \log\frac1{1-\delta}\right)\right)\tau \nonumber\\
    &+ \delta\left( \frac{e^\varepsilon}{1-\delta}  - \frac{1-\delta}{e^\varepsilon} + 4\left( \varepsilon + \log\frac1{1-\delta} +\frac1m \right) \right).
\end{align} \\[-3.5mm]
Comparing to Corollary~\ref{Cor:RE-LDP-bound}, we see that the prefactor of $\tau$ is strictly better in our result. For the overall bounds, we provide a numerical comparison in Figure~\ref{fig:rev-pin-LDP}. For all examples we tried, our new bound performs significantly better (even in the setting $\lambda=m$, which provides a lower bound on the bound given in~\eqref{Eq:dasgupta} since $\lambda \geq m$).

\begin{figure}[!t]
\begin{subfigure}{.45\textwidth}
  \centering  \includegraphics[width=\linewidth]{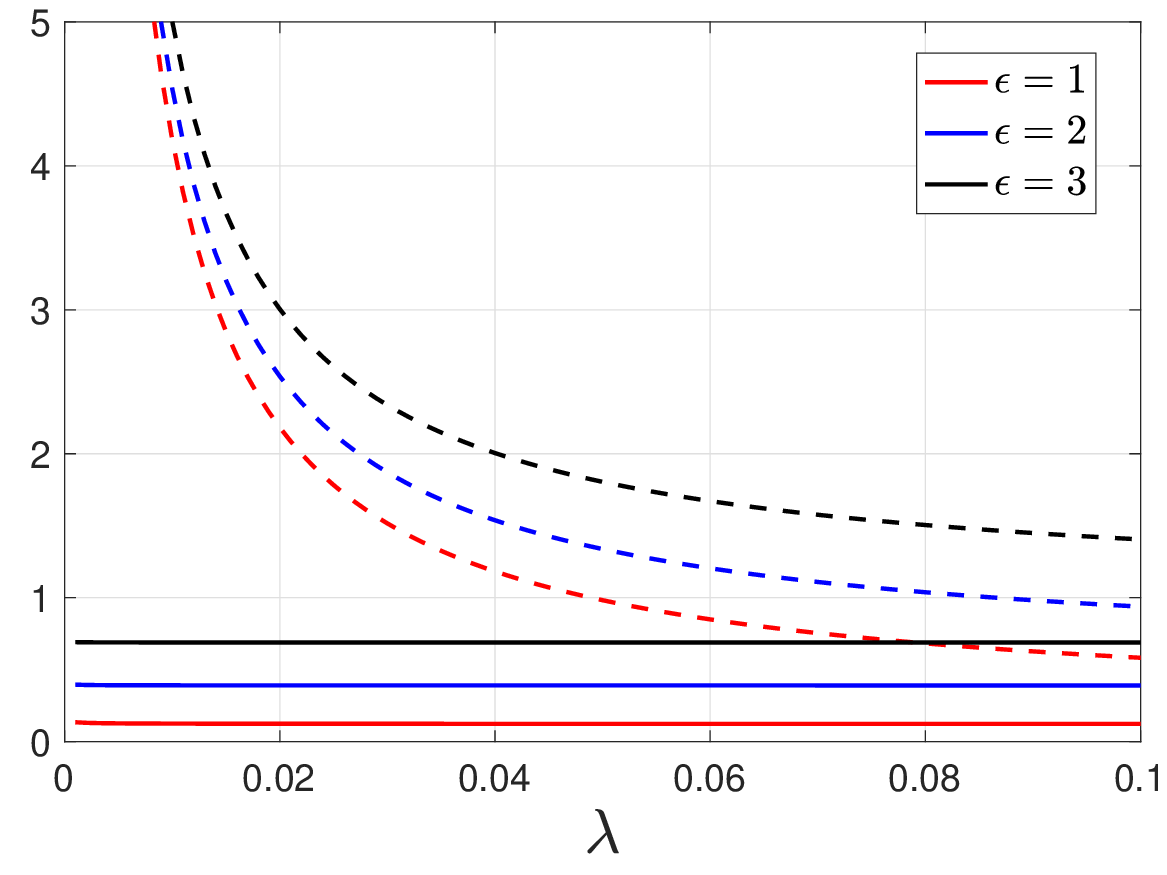}  
  \caption{}
  \label{fig:sub-first}
\end{subfigure}
\begin{subfigure}{.45\textwidth}
  \centering
\includegraphics[width=\linewidth]{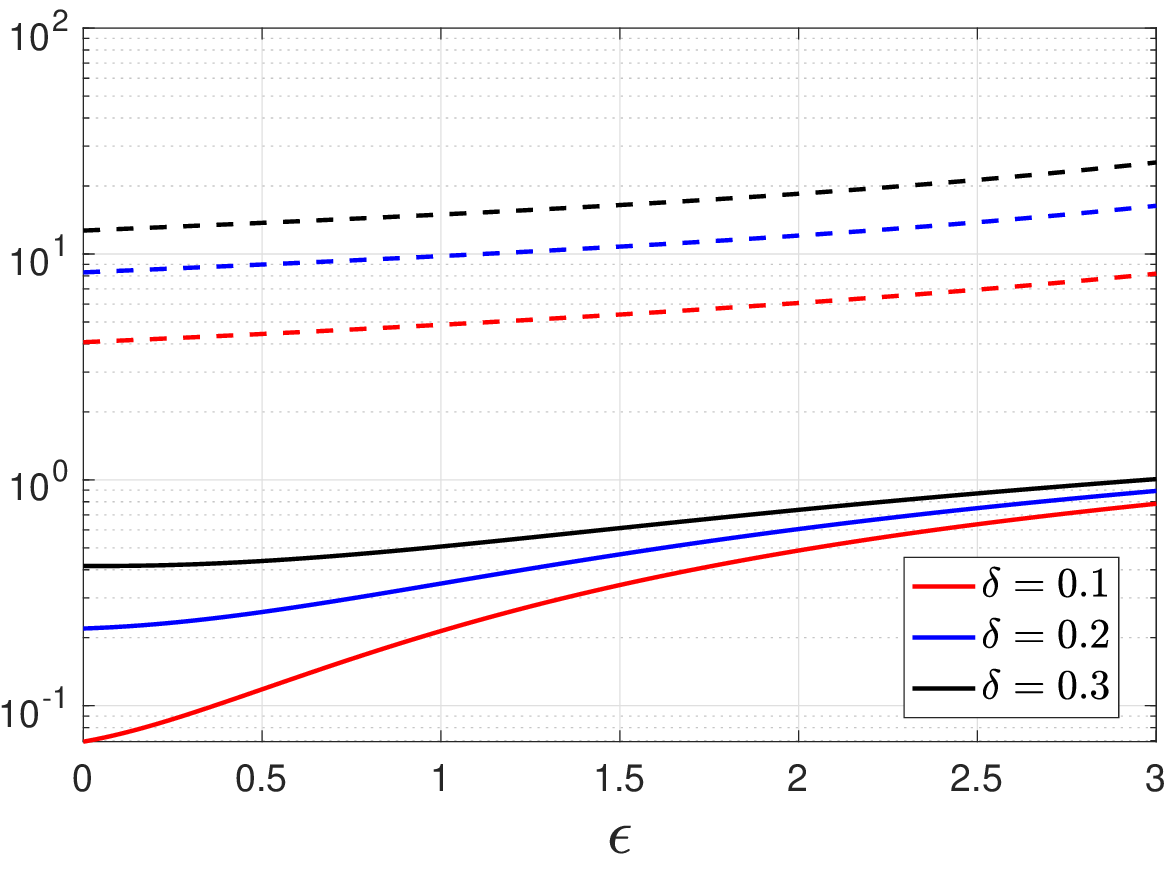}  
  \caption{}
  \label{fig:sub-second}
\end{subfigure}
  \caption{Comparing LDP bounds on the relative entropy. Dashed lines represent Equation~\eqref{Eq:dasgupta} (in particular, a lower bound on that since we chose $\lambda=m$) and solid lines our new bound in Equation~\eqref{Eq:RE-LDP-bound}. (a): Plot over $\lambda$, respectively $m$, for fixed $\varepsilon=\{1,2,3\}, \delta=0.01, \tau=0.25$. (b): Plot over $\varepsilon$ for fixed $\delta=\{0.1,0.2,0.3\},  \lambda=m=0.1, \tau=0.25$.}
    \label{fig:rev-pin-LDP}
\end{figure}
\begin{remark}
    We note that~\Cref{prp:f_div_Contraction} requires $\lambda := \inf_{y \in \cY} A(Q)(y) >0$. This condition is not satisfied by all known $(\varepsilon,\delta)$-LDP mechanisms. However, the given bound is valid for full-support mechanisms such as $k$-randomized response and generalized versions of randomized response (see~\cite{SoKLDPSurvey20}) and also may be relevant depending on the input distribution $Q$. Also, our results generalize the known results for $\delta=0$ in~\cite{zamanlooy2023strong} and improve upon the existing bounds for $\delta >0$ in~\cite{dasgupta2025quantum} as identified from numerical investigations. 
\end{remark}

\section{Conclusion}

In this work, we obtain information contraction bounds that are linear and non-linear for hockey-stick divergence when a private mechanism satisfying $(\varepsilon,\delta)$-LDP is applied as the channel. We either improve or generalize the previously known bounds on these quantities while showcasing instances where our inequalities are in fact tight. For hockey-stick divergence, the non-linear SDPI are tight for some classes of distributions and private mechanisms. It would be an interesting future direction to study whether the upper bounds on $f$-divergences are tight for $\delta >0$. More broadly, it is of interest to study how the parameter $\delta$ impacts the performance of statistical tasks.
We study the quantum extension of these results in~\cite{NGH2025nonHS}.

\bibliographystyle{ieeetr}
\bibliography{lib}

\end{document}